\pdfoutput=1
\documentclass[a4paper]{article}

\usepackage{amsmath,latexsym,amssymb,amsthm}
\usepackage[all]{xy}
\usepackage{bussproofs}
\usepackage{bcprules}

\theoremstyle{plain}
\newtheorem{theorem}{Theorem}
\newtheorem{lemma}[theorem]{Lemma}
\newtheorem{proposition}[theorem]{Proposition}

\theoremstyle{definition}
\newtheorem{definition}{Definition}

\theoremstyle{remark}
\newtheorem{remark}{Remark}
\theoremstyle{plain}

\newcommand*{\institute}[1]{\date{\normalsize #1}}
\newcommand*{\email}[1]{{\ttfamily #1}}

\newcommand*{\nullornot}[3]{\def\@one{#1}\ifx\@one\empty#2\else#3\fi}
\newcommand*{\suffixed}[2]{\nullornot{#1}{#2}{#2_{#1}}}
\newcommand*{\mathopx}[1]{\mathop{#1\protect\mathstrut}\nolimits}
\newcommand*{\abst}[2][]{\mathopx{#1#2.}}
\newcommand*{\mathrelrel}[1]{\mathrel{\,#1\,}}
\newcommand*{\mathconj}[1]{\mathrel{\;#1\;}}
\newcommand*{\ctext}[1]{\mathrm{#1}}
\newcommand*{\const}[2][]{\suffixed{#1}{\ctext{#2}}}
\newcommand*{\operator}[2][]{\mathopx{\ctext{#2}}_{#1}}
\newcommand*{\ftext}[1]{\texttt{#1}}

\newcommand*{\eql}[1][]{\suffixed{#1}{\mathrel{=}}}

\newcommand{\comp}{\mathbin{\circ}}
\newcommand{\blank}{\mathord{-}}

\newcommand*{\norm}[1]{\lvert#1\rvert}
\newcommand*{\vect}[1]{\ensuremath{\overrightarrow{#1}}}
\newcommand*{\tuple}[1]{\langle#1\rangle}
\newcommand*{\eval}[1]{[\![#1]\!]}
\newcommand{\prove}{\mathrelrel{\vdash}}
\newcommand{\csep}{\mathrelrel{\mid}}

\def\imply{\mathbin{\supset}}

\newcommand{\modal}{\mathord{\Box}}

\newcommand*{\mor}[3][]{\mathopx{#1}({#2},{#3})}
\newcommand*{\unit}[1][]{\suffixed{#1}{\eta}}
\newcommand*{\counit}[1][]{\suffixed{#1}{\varepsilon}}
\newcommand*{\mult}[1][]{\suffixed{#1}{\mu}}
\newcommand*{\comult}[1][]{\suffixed{#1}{\delta}}
\newcommand{\type}{\mathrel{:}}
\newcommand*{\uptype}[2]{#2^{#1\!}}
\newcommand*{\var}[2]{{#1}\mathbin{:}{#2}}
\newcommand{\bnfeq}{\mathrel{::=}}
\newcommand{\bnfor}{\mathrel{\mid}}
\newcommand*{\labst}[1]{\abst[\lambda]{#1}}
\newcommand*{\freev}[1]{\operator{FV}(#1)}
\newcommand*{\msubst}[1]{\{#1\}}
\newcommand*{\subst}[2]{\msubst{\nullornot{#1}{#2}{\substx{#1}{#2}}}}
\newcommand*{\substx}[2]{{#2}/{#1}}
\newcommand*{\letin}[2]{\letinlet{{#1}\letinbe{#2}}\letinin}
\newcommand{\letinlet}{\mathopx{\ftext{let}}}
\newcommand{\letinbe}{\mathbin{\ftext{be}}}
\newcommand{\letinin}{\mathrel{\ftext{in}}}
\newcommand*{\linecont}[1][M]{\phantom{#1}}
\newcommand*{\longeqn}[2][M]{\lefteqn{#2}\phantom{#1}}
\newcommand*{\rulename}[1]{\textrm{#1}}
\newcommand*{\logic}[1]{\mbox{\textbf{\textrm{#1}}}}
\newcommand*{\boxin}[3]%
{\mathopx{\mathtt{box}}{\langle#1\rangle}\mathbin{\mathtt{be}}%
{\langle#2\rangle}\mathrel{\mathtt{in}}{#3}}
\newcommand*{\letinmodal}[3]{\letin{\modal{#1}}{#2}{#3}}
\newcommand*{\boxsub}[3]%
{\mathopx{\mathtt{box}}{#3}\mathbin{\mathtt{with}}{#2}%
\mathrel{\mathtt{for}}{#1}}
\newcommand*{\unbox}[1]{\mathopx{\mathtt{unbox}}{#1}}
\newcommand*{\reducto}[1][]{\longrightarrow_{#1}}

\newcommand*{\ureducto}[1][]{\vphantom{\beta}_{/}\!\!\!\longrightarrow_{#1}}
\newcommand*{\dreducto}[1][]{\vphantom{I}^{\backslash}\!\!\!\longrightarrow_{#1}}
\newcommand{\arrow}{\mathord{\imply}}
\newcommand{\tid}{\ctext{id}}
\newcommand*{\betar}[2][]{\beta^{\rulename{#1}}_{#2}}
\newcommand*{\etar}[2][]{\eta^{\rulename{#1}}_{#2}}
\newcommand*{\idr}[1]{\tid_{#1}}
\newcommand*{\cps}[1]{\eval{#1}}
\newcommand*{\cpsv}[1]{\overline{#1}}
\newcommand*{\tcps}[1]{\overline{#1}}
\newcommand*{\cpsx}[1]{\eval{#1}'}
\newcommand*{\tcpsx}[1]{\overline{#1}'}
\newcommand*{\mcps}[1]{\varPhi(#1)}
\newcommand*{\icps}[1]{\varPsi(#1)}
\newcommand{\return}{\const{R}}
\newcommand*{\ceil}[1]{\lceil#1\rceil}
\newcommand*{\floor}[1]{\lfloor#1\rfloor}
\newcommand*{\ceilx}[1]{\langle\!|#1|\!\rangle}
\newcommand*{\floorx}[1]{(\!|#1|\!)}

\title{Calculi for Intuitionistic Normal Modal Logic
\footnote{This paper was reported at PPL 2007.
The results of this paper have already been published as
\texttt{DOI:10.11309/jssst.25.1\_167} written in Japanese.}}
\author{Yoshihiko Kakutani}
\institute{Department of Information Science, University of Tokyo \\
\email{kakutani@is.s.u-tokyo.ac.jp}}

\begin{document}

\maketitle

\begin{abstract}
 This paper provides a call-by-name and a call-by-value term calculus,
 both of which have a Curry-Howard correspondence to
 the box fragment of the intuitionistic modal logic \logic{IK}\@.
 The strong normalizability and the confluency of the calculi are shown.
 Moreover, we define a CPS transformation from
 the call-by-value calculus to the call-by-name calculus,
 and show its soundness and completeness.
\end{abstract}

\section{Introduction}\label{SS:intro}

It is well-known that the intuitionistic propositional logic
exactly corresponds to the simply typed $\lambda$-calculus:
formulae as types and proofs as terms.
Such a correspondence is called a Curry-Howard correspondence
after \cite{Howard80:FTNC}\@.
A Curry-Howard correspondence enables us to study
an equality on proofs of a logic computationally.
Though Curry-Howard correspondences for
higher-order and predicate logics were provided in \cite{Barendregt92:LCT},
we investigate only propositional logics in this paper.
The aim of this study is to give a proper calculus
that have a Curry-Howard correspondence with a modal logic.

Modal logics have a long history and are now widely studied
both theoretically and practically.
Especially, studies about Kripke semantics~\cite{Kripke63:SAMLINPL}
of modal logics are quite active.
Curry-Howard correspondences of modal logics are, however,
less studied except for linear logics~\cite{Girard87:LL}\@.
(In fact, exponentials of linear logics are a kind of \logic{S4} modality.)
Since \logic{K} is known to be the simplest modal logic,
first we focus the intuitionistic modal logic \logic{IK}\@.
A difficulty of a calculus for \logic{K} is
lack of acknowledged models.
Because a model of the modality in call-by-name \logic{S4}
is acknowledged as a monoidal comonad,
a model of the modality in \logic{K} should be
a generalization of a monoidal comonad.
This paper defines a call-by-name calculus,
which is called the $\lambda\modal$-calculus,
based on a categorical model proposed
by Bellin et al.\ in \cite{BellinPaiva01:ECHCBCML}\@.
Another difficulty is a problem about natural deductions of modal logics
pointed out in \cite{Prawitz65:ND}\@.
A solution of the problem in \logic{IS4} is found
in \cite{Barber96:DILL},
but it cannot be applied to \logic{IK}\@.
The formulation of \cite{BellinPaiva01:ECHCBCML} and this paper
is a natural deduction style, and solves this problem.

On the other hand, studies on Curry-Howard correspondences
for modal logics, especially \logic{IS4}, are applied to
staged computations and information flow analysis
(e.g., \cite{DaviesPfenning01:MASC}, \cite{MiyamotoIgarashi04:MFSIF})
in the field of programming languages.
Since our $\lambda\modal$-calculus can be extended easily to
\logic{IT}, \logic{IK4}, \logic{IS4}, and so on,
this work is expected to contribute such programming language matters.

This paper provides not only a call-by-name calculus
but also a call-by-value one.
A call-by-value calculus is usually defined by a CPS transform,
which is originally introduced
by \cite{Fischer72:LCS} and \cite{Plotkin75:CNCVLC};
for example, a call-by-value control operator
is defined by a CPS transform in \cite{FelleisenFriedman87:STSC}\@.
In \cite{SabryFelleisen93:RPCPS}, Sabry and Felleisen showed that
the $\lambda_{\textrm{c}}$-calculus~\cite{Moggi89:CLCM} is
sound and complete for CPS semantics.
We give the call-by-value $\lambda\modal$-calculus as
an extension of the $\lambda_{\textrm{c}}$-calculus.
Moreover, we define a CPS transformation from
the call-by-value $\lambda\modal$-calculus
to the call-by-name $\lambda\modal$-calculus.
The soundness and completeness for the CPS semantics
are shown along the line of \cite{SabryFelleisen93:RPCPS}\@.

\section{Call-by-Name Calculus}\label{SS:cbn}

First, we remark special notations used in this paper.
We use a notation ``$\vect{M}$'' for
a sequence of meta-variables ``$M_1,\ldots,M_n$''
including the empty sequence.
Hence, an expression ``$\vect{M},\vect{N}$'' stands for
the concatenation of $\vect{M}$ and $\vect{N}$\@.
For a unary operator $\Phi(\blank)$,
we write ``$\Phi(\vect{M})$'' for
the sequence ``$\Phi(M_1),\ldots,\Phi(M_n)$''\@.
We use also ``$\vect{N}(\labst{\vect{x}}{M})$''
as an abbreviation for
``$N_1(\labst{x_1}{\cdots N_n(\labst{x_n}{M})\cdots})$''\@.

A hole in a context is represented by ``$\blank$'' in this paper.
For a context $C$, ``$C[M]$'' denotes the result of
filling holes in $C$ with $M$ as usual.

\begin{definition}
 Types $\sigma$ and terms $M$ of
 the call-by-name $\lambda\modal$-calculus are defined as follows:
 \begin{align*}
  \sigma &\bnfeq p \bnfor \sigma\imply\sigma
  \bnfor \modal{\sigma} ,\\
  M &\bnfeq c \bnfor x \bnfor \labst{\uptype{\sigma}x}{M} \bnfor MM
  \bnfor \boxin{\uptype{\sigma}x,\ldots,\uptype{\sigma}x}{M,\ldots,M}{M} ,
 \end{align*}
 where $p$, $c$, and $x$ range over type constants,
 constants, and variables, respectively.
 Free variables of\/ $\boxin{\vect{x}}{\vect{N}}{M}$
 are free variables of $\vect{N}$\@.
 The typing rules are given in Figure~\ref{FIG:type}\@.
 The reduction rules are given in Figure~\ref{FIG:cbn}\@.
 Define $\const{n}$ as the set
 $\{\betar{\arrow},\etar{\arrow},\idr{\modal},\betar{\modal}\}$\@.
\end{definition}

\begin{figure}[t]
 \infrule{}{\varGamma \prove \uptype{\tau}c \type{\tau}}
 \infrule{}{\varGamma,\var{x}{\tau},\varGamma' \prove x \type{\tau}}
 \infrule{\varGamma,\var{x}{\sigma} \prove M \type{\tau}}
 {\varGamma \prove \labst{\uptype{\sigma}x}{M} \type{\sigma\imply \tau}}
 \infrule{\varGamma \prove M \type{\sigma\imply\tau}
  \andalso \varGamma \prove N \type{\sigma}}
 {\varGamma \prove MN \type{\tau}}
 \infrule
 {\var{x_1}{\sigma_1},\ldots,\var{x_n}{\sigma_n} \prove M \type{\tau}
  \andalso \varGamma \prove N_1 \type{\modal{\sigma_1}}
  \andalso\cdots\andalso \varGamma \prove N_n \type{\modal{\sigma_n}}}
 {\varGamma \prove \boxin{
   \uptype{\sigma_1}x_1,\ldots,\uptype{\sigma_n}x_n}{N_1,\ldots,N_n}{M}
  \type{\modal{\tau}}}
 \caption{Typing rules of $\lambda\modal$-calculus}
 \label{FIG:type}
\end{figure}

\begin{figure}[t]
 \begin{align*}
  &(\labst{x}{M})N \reducto[\betar{\arrow}] M\subst{x}{N} \\
  &\labst{x}{Mx} \reducto[\etar{\arrow}] M &&x \not\in \freev{M} \\
  &\boxin{x}{M}{x} \reducto[\idr{\modal}] M \\
  &\boxin{\vect{w},x,\vect{z}}
  {\vect{P},\boxin{\vect{y}}{\vect{L}}{N},\vect{Q}}{M} \\
  &\linecont \reducto[\betar{\modal}]
  \boxin{\vect{w},\vect{y},\vect{z}}{\vect{P},\vect{L},\vect{Q}}
  {M\subst{x}{N}}
  &&\norm{\vect{w}} = \norm{\vect{P}}
 \end{align*}
 \caption{Call-by-name reductions of $\lambda\modal$-calculus}
 \label{FIG:cbn}
\end{figure}

We remark that all free variables of $M$ are
included by $\{\vect{x}\}$
if $\boxin{\vect{x}}{\vect{N}}{M}$ is typable.

The $\lambda\modal$-calculus has essentially the same syntax
as \cite{BellinPaiva01:ECHCBCML}\@.
Hence, one can see that
our calculus corresponds to the intuitionistic modal logic.
The calculus can be regarded as a natural deduction
by forgetting terms.
Our logic is equivalent to the $\arrow\modal$-fragment
of the usual intuitionistic modal logic \logic{IK}
with respect to provability.
Let \logic{IK} be an intuitionistic Hilbert system
with the axiom
$\modal(\sigma\imply\tau)\imply\modal{\sigma}\imply\modal{\tau}$
and the box inference rule.
The axiom is validated in our calculus as the term
\begin{gather*}
 \prove \labst{f}\labst{x}{\boxin{f',x'}{f,x}{f'x'}}
 \type{\modal(\sigma\imply\tau)\imply\modal{\sigma}\imply\modal{\tau}} .
\end{gather*}
The box rule is simulated as
\begin{gather*}
 \AxiomC{$\prove M \type{\tau}$}
 \UnaryInfC{$\prove \boxin{}{}{M} \type{\modal{\tau}}$}
 \bottomAlignProof
 \DisplayProof .
\end{gather*}
Conversely, the typing rule of the $\lambda\modal$-calculus
is simulated by \logic{IK}:
\begin{gather*}
 \AxiomC{$\sigma_1,\ldots,\sigma_n \prove \tau$}
 \UnaryInfC{$\prove \sigma_1\imply\cdots\imply\sigma_n\imply\tau$}
 \UnaryInfC{$\prove \modal(\sigma_1\imply\cdots\imply\sigma_n\imply\tau)$}
 \UnaryInfC{$\varGamma \prove
 \modal(\sigma_1\imply\cdots\imply\sigma_n\imply\tau)$}
 \UnaryInfC{$\varGamma \prove \modal{\sigma_1}\imply
 \modal(\sigma_2\imply\cdots\imply\sigma_n\imply\tau)$}
 \AxiomC{$\varGamma \prove \modal{\sigma_1}$}
 \BinaryInfC{$\varGamma \prove
 \modal(\sigma_2\imply\cdots\imply\sigma_n\imply\tau)$}
 \UnaryInfC{$\vdots$}
 \UnaryInfC{$\varGamma \prove \modal(\sigma_n\imply\tau)$}
 \UnaryInfC{$\varGamma \prove \modal{\sigma_n}\imply\modal{\tau}$}
 \AxiomC{$\varGamma \prove \modal{\sigma_n}$}
 \BinaryInfC{$\varGamma \prove \modal{\tau}$}
 \bottomAlignProof
 \DisplayProof .
\end{gather*}

According to this encoding, it is not trivial
whether an exchange rule commutes with a box operation.
Hence, we distinguish two terms,
$\boxin{x,y}{N,L}{M}$ and $\boxin{y,x}{L,N}{M}$,
in the $\lambda\modal$-calculus,
although it is common to consider proofs up to exchanges.
Commutativity with exchanges requires another axiom,
\emph{symmetricity}, given in Section~\ref{SS:S4}\@.

\begin{remark}
 The typing rules of our calculus are the same as
 those of Bellin et al.'s~\cite{BellinPaiva01:ECHCBCML},
 but reductions are essentially different.
 In \cite{BellinPaiva01:ECHCBCML}, they addresses
 natural deduction style formulation and categorical semantics,
 but not a term calculus itself,
 so their calculus has room for improvement.
 Differences between Bellin et al.'s calculus and
 our $\lambda\modal$-calculus are the following.
 \begin{itemize}
  \item The first reduction of their calculus corresponds to
	a special case of $\reducto[\betar{\modal}]$\@.
  \item The direction of the first reduction is opposite
	to $\reducto[\betar{\modal}]$:
	our reduction merges adjacent two boxes into one box,
	while their reduction splits a box into two boxes.
  \item The second reduction of their calculus cannot be applied
	to any typable term.
  \item Their calculus does not have a reduction
	corresponding to $\reducto[\idr{\modal}]$\@.
 \end{itemize}
 Though Bellin et al.'s calculus does not have
 the semantic completeness,
 it has syntax for a diamond property.
 Intuitionistic characterization of a diamond property
 is not obvious,
 but the author has observed a diamond property
 in the classical modal logic \logic{K}
 in \cite{Kakutani07:CNCVNML}\@.
\end{remark}

For a set of labels $X$, we write $\reducto[X]$ as
a reduction whose label is a member of $X$,
and $\eql[X]$ as the reflexive transitive symmetric closure
of $\reducto[X]$\@.
We also use $\equiv$ for the $\alpha$-equivalence.

We can easily check the subject reduction theorem for this calculus.

\begin{proposition}
 If $\varGamma \prove M \type{\tau}$ and $M \reducto[\const{n}] N$ hold,
 then $\varGamma \prove N \type{\tau}$ holds.
\end{proposition}

Other important properties, the strong normalizability and
the confluency, also hold.

\begin{proposition}
 The call-by-name $\lambda\modal$-calculus is strongly normalizable
 with respect to $\reducto[\const{n}]$\@.
\end{proposition}

\begin{proof}
 Define the transformation $\ceil{\blank}$
 to the simply typed $\lambda$-calculus by
 \begin{align*}
  \ceil{\boxin{\vect{x}}{\vect{N}}{M}} =
  &\labst{k}{\ceil{\vect{N}}(\labst{\vect{x}}{k\ceil{M}})} .
 \end{align*}
 Then, $M \type{\tau}$ implies $\ceil{M} \type{\ceil{\tau}}$
 if we define the type transformation $\ceil{\blank}$
 by $\ceil{\modal{\tau}} = (\ceil{\tau}\imply p)\imply p$\@.
 One can see that
 $\ceil{\boxin{x}{M}{x}} \reducto[\etar{\arrow}]^+ \ceil{M}$ and
 \begin{align*}
  \longeqn{\ceil{\boxin{\vect{w},x,\vect{z}}
  {\vect{P},\boxin{\vect{y}}{\vect{L}}{N},\vect{Q}}{M}}} \\
  &\equiv \labst{k}{\ceil{\vect{P}}(\labst{\vect{w}}
  {(\labst{h}{\ceil{\vect{L}}(\labst{\vect{y}}{h\ceil{N}})})(\labst{x}
  {\ceil{\vect{Q}}(\labst{\vect{z}}{k\ceil{M}})})})} \\
  &\reducto[\betar{\arrow}] \labst{k}{\ceil{\vect{P}}(\labst{\vect{w}}
  {\ceil{\vect{L}}(\labst{\vect{y}}
  {(\labst{x}{\ceil{\vect{Q}}(\labst{\vect{z}}{k\ceil{M}})})
  \ceil{N}})})} \\
  &\reducto[\betar{\arrow}] \labst{k}{\ceil{\vect{P}}(\labst{\vect{w}}
  {\ceil{\vect{L}}(\labst{\vect{y}}{\ceil{\vect{Q}}(\labst{\vect{z}}
  {k\ceil{M}\subst{x}{\ceil{N}}})})})} \\
  &\equiv \labst{k}{\ceil{\vect{P}}(\labst{\vect{w}}
  {\ceil{\vect{L}}(\labst{\vect{y}}{\ceil{\vect{Q}}(\labst{\vect{z}}
  {k\ceil{M\subst{x}{N}}})})})} \\
  &\equiv \ceil{\boxin{\vect{w},\vect{y},\vect{z}}
  {\vect{P},\vect{L},\vect{Q}}{M\subst{x}{N}}}
 \end{align*}
 hold.
 Because the simply typed $\lambda$-calculus is SN
 w.r.t.\ $\reducto[\betar{\arrow},\etar{\arrow}]$
 (e.g., q.v.\ \cite{GirardTaylor89:PT}),
 the call-by-name $\lambda\modal$-calculus is SN\@.
\end{proof}

We note here that the strong normalization theorem was proved
via a different calculus in \cite{Abe07:CMPFOPL}\@.

\begin{proposition}
 $\reducto[\const{n}]$ is confluent.
\end{proposition}

\begin{proof}
 By Newman's lemma~\cite{Newman42:TCDE},
 it is sufficient to check the local confluency.
 The call-by-name $\lambda\modal$-calculus has essentially
 four kinds of critical pairs other than
 pairs of the $\lambda$-calculus:
 \begin{align*}
  &\linecont \ureducto[\idr{\modal}]
  \boxin{\vect{y}}{\vect{L}}{N} \\
  &\boxin{x}{\boxin{\vect{y}}{\vect{L}}{N}}{x} \\
  &\linecont \dreducto[\betar{\modal}]
  \boxin{\vect{y}}{\vect{L}}{N} ,\\
  &\linecont \ureducto[\idr{\modal}] \boxin{x}{N}{M} \\
  &\boxin{x}{\boxin{y}{N}{y}}{M} \\
  &\linecont \dreducto[\betar{\modal}]
  \boxin{y}{N}{M\subst{x}{y}} ,\\
  &\linecont \ureducto[\betar{\modal}]
  \boxin{y}{\boxin{\vect{z}}{\vect{P}}{L}}{M\subst{x}{N}} \\
  &\boxin{x}{\boxin{y}{\boxin{\vect{z}}{\vect{P}}{L}}{N}}{M} \\
  &\linecont \dreducto[\betar{\modal}]
  \boxin{x}{\boxin{\vect{z}}{\vect{P}}{N\subst{y}{L}}}{M} , \\
  &\linecont \ureducto[\betar{\modal}]
  \boxin{\vect{y},x'}{\vect{L},
  \boxin{\vect{y'}}{\vect{L'}}{N'}}{M\subst{x}{N}} \\
  &\boxin{x,x'}{\boxin{\vect{y}}{\vect{L}}{N},
  \boxin{\vect{y'}}{\vect{L'}}{N'}}{M} \\
  &\linecont \dreducto[\betar{\modal}]
  \boxin{x,\vect{y'}}{\boxin{\vect{y}}{\vect{L}}{N},
  \vect{L'}}{M\subst{x'}{N'}} .
 \end{align*}
 It is easily shown that all the pairs are joinable.
\end{proof}

Last, we mention the subformula property of this calculus.

\begin{theorem}
 A normal form in the call-by-name $\lambda\modal$-calculus
 has the subformula property.
\end{theorem}

\begin{proof}
 By induction on construction of terms.
 If $\boxin{\vect{x}}{\vect{N}}{M}$ is a normal form,
 then $M$ is a normal form and each $N_i$ has
 a form $yL_1\cdots L_m$\@.
 Therefore, the subformula property holds in this case
 by the induction hypothesis.
 Other cases are just the same as
 the simply typed $\lambda$-calculus.
\end{proof}

A characterization of the $\lambda\modal$-calculus
by a standard translation into the predicate logic
is given by Abe in \cite{Abe07:CMPFOPL}\@.
Since our motivation arises from logics and categorical semantics,
computational meaning of the calculus still remains to be studied.
We believe the following discussions are helpful.

Because the logic \logic{IK} is weaker than the logic \logic{IS4},
the $\lambda\modal$-calculus is expected to be a subcalculus
of a calculus for \logic{IS4}\@.
A method for extending the $\lambda\modal$-calculus
to \logic{IS4} is discussed in Section~\ref{SS:S4}\@.
Through this approach, computational analyses of \logic{IS4} calculi
might be applied to the $\lambda\modal$-calculus.

Another approach to understand computational meaning
of the $\lambda\modal$-calculus
is to investigate a relation to monads.
It is remarkable that the transformation $\ceil{\blank}$
mentioned in the proof of Proposition~2 preserves the equality.
It means that $\modal$ in the $\lambda\modal$-calculus
can be interpreted as a continuation monad in the $\lambda$-calculus.
In fact, such a transformation exists for any strong monad
because a strong monad is a lax monoidal endofunctor.
Hence, we can conclude that the $\lambda\modal$-calculus
includes an abstract setting of strong monads.
In \cite{McBridePaterson07:APE},
McBride and Paterson have studied a structure
abstracting a strong monad.
It must be strongly related to our calculus
though their formulation
has a tensorial strength with respect to cartesian products.

\section{Call-by-Value Calculus}\label{SS:cbv}

\begin{definition}
 Types $\sigma$, terms $M$, values $V$,
 simple evaluation contexts $C$, and evaluation contexts $E$
 of the call-by-value $\lambda\modal$-calculus are defined as follows:
 \begin{align*}
  \sigma &\bnfeq p \bnfor \sigma\imply\sigma
  \bnfor \modal{\sigma} ,\\
  M &\bnfeq c \bnfor x \bnfor \labst{\uptype{\sigma}x}{M} \bnfor MM
  \bnfor \boxin{\uptype{\sigma}x,\ldots,\uptype{\sigma}x}{M,\ldots,M}{M} ,\\
  V &\bnfeq c \bnfor x \bnfor \labst{\uptype{\sigma}x}{M}
  \bnfor \boxin{\uptype{\sigma}x,\ldots,\uptype{\sigma}x}{V,\ldots,V}{M} ,\\
  C &\bnfeq \blank M \bnfor V\blank
  \bnfor \boxin{\uptype{\sigma}x,\ldots,\uptype{\sigma}x}
  {V,\ldots,V,\blank,M,\ldots,M}{M} ,\\
  E &\bnfeq \blank \bnfor C[E] .
 \end{align*}
 The typing rules are just the same as the call-by-name.
 The reduction rules are given in Figure~\ref{FIG:cbv}\@.
 Define $\const{v}$ as the set
 $\{\idr{\arrow},\betar[v]{\arrow},\etar[v]{\arrow},\rulename{lift},
 \rulename{flat},\betar{\Omega},\idr{\modal},\betar[v]{\modal}\}$\@.
\end{definition}

\begin{figure}[t]
 \begin{align*}
  &V,W \type{\rulename{value}} \\
  &C \type{\rulename{simple evaluation context}} \\
  &E \type{\rulename{evaluation context}} \\
  &(\labst{x}{x})M \reducto[\idr{\arrow}] M \\
  &(\labst{x}{M})V \reducto[{\betar[v]{\arrow}}] M\subst{x}{V} \\
  &\labst{x}{Vx} \reducto[{\etar[v]{\arrow}}] V &&x \not\in \freev{V} \\
  &C[(\labst{x}{M})N] \reducto[\rulename{lift}] (\labst{x}{C[M]})N \\
  &C[yM] \reducto[\rulename{flat}] (\labst{x}{C[x]})(yM)
  &&C \neq V\blank \\
  &(\labst{x}{E[yx]})M \reducto[\betar{\Omega}] E[yM]
  && x \not\in \freev{E[y]} \\
  &\boxin{x}{M}{x} \reducto[\idr{\modal}] M \\
  &\boxin{\vect{w},x,\vect{z}}
  {\vect{W},\boxin{\vect{y}}{\vect{N}}{V},\vect{P}}{M} \\
  &\linecont \reducto[{\betar[v]{\modal}}]
  \boxin{\vect{w},\vect{y},\vect{z}}{\vect{W},\vect{N},\vect{P}}
  {M\subst{x}{V}}
  &&\norm{\vect{w}} = \norm{\vect{W}}
 \end{align*}
 \caption{Call-by-value reductions of $\lambda\modal$-calculus}
 \label{FIG:cbv}
\end{figure}

\begin{proposition}
 If $\varGamma \prove M \type{\tau}$ and $M \reducto[\const{v}] N$ hold,
 then $\varGamma \prove N \type{\tau}$ holds.
\end{proposition}

Since the definition of terms and the typing rules are the same
as those of the call-by-name calculus,
also the call-by-value $\lambda\modal$-calculus
corresponds to \logic{IK}\@.
In order to define CPS semantics, however,
we restrict terms as follows:
\begin{align*}
 M &\bnfeq c \bnfor x \bnfor \labst{\uptype{\sigma}x}{M} \bnfor MM
 \bnfor \boxin{\uptype{\sigma}x,\ldots,\uptype{\sigma}x}{M,\ldots,M}{V} .
\end{align*}
These terms are closed under call-by-value reductions
because values are closed under substitutions.
Hence, we can say that the full call-by-value calculus is
a conservative extension of the restricted version.
In the rest of this section (and the first half of the next section),
we focus on this restricted calculus.

Our call-by-value $\lambda\modal$-calculus is an extension of
Sabry and Felleisen's calculus in \cite{SabryFelleisen93:RPCPS}\@.
As mentioned in \cite{SabryFelleisen93:RPCPS},
it is equivalent to the $\lambda_{\textrm{c}}$-calculus~\cite{Moggi89:CLCM},
which is acknowledged as a call-by-value language,
with respect to equalities.
We give CPS semantics of the call-by-value $\lambda\modal$-calculus
and show the soundness and completeness
along the line of \cite{SabryFelleisen93:RPCPS}\@.

\begin{definition}
 The CPS transformation $\cps{\blank}$
 from the call-by-value $\lambda\modal$-calculus
 to the call-by-name $\lambda\modal$-calculus is defined
 by Fig~\ref{FIG:cps}\@.
 We write $\mcps{M,K}$ for the administrative normal form of $\cps{M}K$\@.
\end{definition}

\begin{figure}[t]
 \begin{align*}
  \tcps{p} &= p \\
  \tcps{\sigma\imply\tau}
  &= (\tcps{\tau}\imply\return)\imply\tcps{\sigma}\imply\return \\
  \tcps{\modal{\sigma}} &= \modal{\tcps{\sigma}} \\
  \cpsv{x} &= x \\
  \cpsv{c} &= c \\
  \cpsv{\labst{x}{M}} &= \labst{k}\labst{x}{\cps{M}k} \\
  \cpsv{\boxin{\vect{x}}{\vect{U}}{V}}
  &= \boxin{\vect{x}}{\cpsv{\vect{U}}}{\cpsv{V}} \\
  \cps{x} &= \labst{k}{k\cpsv{x}} \\
  \cps{c} &= \labst{k}{k\cpsv{c}} \\
  \cps{\labst{x}{M}} &= \labst{k}{k(\cpsv{\labst{x}{M}})} \\
  \cps{MN} &= \labst{k}{\cps{M}(\labst{y}{\cps{N}(yk)})} \\
  \cps{\boxin{\vect{x}}{\vect{M}}{V}}
  &= \labst{k}{\cps{\vect{M}}(\labst{\vect{y}}
  {k(\cpsv{\boxin{\vect{x}}{\vect{y}}{V}})})}
 \end{align*}
 \caption{CPS transformation with $\modal$}
 \label{FIG:cps}
\end{figure}

\begin{proposition}
 If $\var{x_1}{\sigma_1},\ldots,\var{x_n}{\sigma_n} \prove M \type{\tau}$ holds,
 $\var{x_1}{\tcps{\sigma_1}},\ldots,\var{x_n}{\tcps{\sigma_n}} \prove
 \cps{M} \type{(\tcps{\tau}\imply\return)\imply\return}$ holds.
\end{proposition}

\begin{definition}
 The CPS language is defined as a subcalculus of
 the call-by-name $\lambda\modal$-calculus:
 \begin{align*}
  V &\bnfeq c \bnfor x \bnfor \labst{k}{K}
  \bnfor \boxin{x,\ldots,x}{V,\ldots,V}{V} ,\\
  K &\bnfeq k \bnfor \labst{x}{A} \bnfor VK ,\\
  A &\bnfeq KV \bnfor (\labst{k}{A})K .
 \end{align*}
 The transformation $\icps{\blank}$ from the CPS language
 to the call-by-value $\lambda\modal$-calculus
 is defined by Fig~\ref{FIG:invcps}\@.
\end{definition}

\begin{figure}[t]
 \begin{align*}
  \icps{c} &= c \\
  \icps{x} &= x \\
  \icps{\labst{k}{k}} &= \labst{x}{x} \\
  \icps{\labst{k}\labst{x}{A}} &= \labst{x}\icps{A} \\
  \icps{\labst{k}{VK}} &= \labst{x}\icps{VKx} \\
  \icps{\boxin{\vect{x}}{\vect{U}}{V}}
  &= \boxin{\vect{x}}{\icps{\vect{U}}}{\icps{V}} \\
  \icps{k} &= \blank \\
  \icps{\labst{x}{A}} &= (\labst{x}{\icps{A}})\blank \\
  \icps{cK} &= \icps{K}[c\blank] \\
  \icps{xK} &= \icps{K}[x\blank] \\
  \icps{(\labst{k}{H})K} &= \icps{H\subst{k}{K}} \\
  \icps{KV} &= \icps{K}[\icps{V}] \\
  \icps{(\labst{k}{A})K} &= \icps{A\subst{k}{K}}
 \end{align*}
 \caption{Inverse of CPS transformation}
 \label{FIG:invcps}
\end{figure}

\begin{proposition}
 The CPS language is closed under $\reducto[\const{n}]$\@.
\end{proposition}

The following lemma is the core of the soundness and completeness.
An outline of the proof is just the same
as \cite{SabryFelleisen93:RPCPS}'s.

\begin{lemma}
 \begin{enumerate}
  \item  $M \reducto[\rulename{lift},\rulename{flat}] N$ implies
	 $\mcps{M,k} \equiv \mcps{N,k}$\@.
  \item  $M \reducto[\idr{\arrow},{\betar[v]{\arrow}},{\betar[v]{\arrow}},
	 \betar{\Omega},\idr{\modal},{\betar[v]{\modal}}] N$
	 implies $\mcps{M,k} \reducto[\const{n}]^+ \mcps{N,k}$\@.
  \item  $M \reducto[\const{n}] N$ implies
	 $\icps{M} \reducto[\const{v}]^\ast \icps{N}$\@.
  \item $M \reducto[\rulename{lift},\rulename{flat}]^\ast \icps{\mcps{M,k}}$\@.
 \end{enumerate}
\end{lemma}

\begin{theorem}
 For $\lambda\modal$-terms $M$ and $N$,
 $M \eql[\const{v}] N$ holds if and only if
 $\cps{M} \eql[\const{n}] \cps{N}$ holds.
\end{theorem}

The lemma helps us to prove
the strongly normalizing property and the confluency
of the call-by-value $\lambda\modal$-calculus too.

\begin{proposition}
 The call-by-value $\lambda\modal$-calculus is strongly normalizable
 with respect to $\reducto[\const{v}]$\@.
\end{proposition}

\begin{proof}
 There is no infinite sequence of $\reducto[\rulename{lift}]$
 and $\reducto[\rulename{flat}]$.
 Therefore, if there is an infinite reduction sequence
 in the call-by-value $\lambda\modal$-calculus,
 there is an infinite reduction sequence
 in the call-by-name calculus via $\mcps{\blank,k}$\@.
\end{proof}

\begin{proposition}
 $\reducto[\const{v}]$ is confluent.
\end{proposition}

\begin{proof}
 Although the confluency can be shown directly,
 we prove it using the lemma and the confluency of
 the call-by-name $\lambda\modal$-calculus.
 Assume $M \reducto[\const{v}]^\ast N_1$ and
 $M \reducto[\const{v}]^\ast N_2$\@.
 Since $\mcps{M,k} \reducto[\const{n}]^\ast \mcps{N_j,k}$,
 there is a term $L$ such that $\mcps{N_j,k} \reducto[\const{n}]^\ast L$\@.
 $\icps{L}$ is an evidence of confluence.
\end{proof}

\section{Other Formulations of Call-by-Value}\label{SS:annex}

Although it has been shown that
the call-by-value $\lambda\modal$-calculus
has expected properties,
we can propose another call-by-value axiomatization
following \cite{Moggi88:CLCM}\@.

\begin{definition}
 Define the computational $\lambda\modal$-calculus
 by adding the new syntax $\letin{x}{N}{M}$ to
 the syntax of the call-by-value $\lambda\modal$-calculus.
 The reduction rules of the computational $\lambda\modal$-calculus
 are given in Figure~\ref{FIG:cbv2}\@.
 Define $\const{c}$ as the set
 $\{\idr{\mathtt{let}},\betar[v]{\mathtt{let}},
 \betar[v]{\arrow},\etar[v]{\arrow},\ctext{comp},\ctext{let},
 \idr{\modal},\betar[v]{\modal}\}$\@.
\end{definition}

\begin{figure}[t]
 \begin{align*}
  &V,W \type{\text{value}} \\
  &A \type{\text{non-value}} \\
  &C \type{\text{simple evaluation context}} \\
  &\letin{x}{M}{x} \reducto[\idr{\mathtt{let}}] M \\
  &\letin{x}{V}{M} \reducto[{\betar[v]{\mathtt{let}}}] M\subst{x}{V} \\
  &(\labst{x}{M})V \reducto[{\betar[v]{\arrow}}] M\subst{x}{V} \\
  &\labst{x}{Vx} \reducto[{\etar[v]{\arrow}}] V &&x \not\in \freev{V} \\
  &\letin{x}{(\letin{y}{L}{N})}{M} \\
  &\linecont \reducto[\ctext{comp}] \letin{y}{L}{\letin{x}{N}{M}}
  &&y \not\in \freev{M} \\
  &C[A] \reducto[\ctext{let}] \letin{x}{A}{C[x]} \\
  &\boxin{x}{M}{x} \reducto[\idr{\modal}] M \\
  &\boxin{\vect{w},x,\vect{z}}
  {\vect{W},\boxin{\vect{y}}{\vect{N}}{V},\vect{P}}{M} \\
  &\linecont \reducto[{\betar[v]{\modal}}]
  \boxin{\vect{w},\vect{y},\vect{z}}{\vect{W},\vect{N},\vect{P}}{M\subst{x}{V}}
  &&\norm{\vect{w}} = \norm{\vect{W}}
 \end{align*}
 \caption{Computational reductions of $\lambda\modal$-calculus}
 \label{FIG:cbv2}
\end{figure}

\begin{proposition}
 If $\varGamma \prove M \type{\tau}$ and $M \reducto[\const{c}] N$ hold,
 then $\varGamma \prove N \type{\tau}$ holds.
\end{proposition}

It is easily seen that the computational $\lambda\modal$-calculus
is equivalent to the previous call-by-value $\lambda\modal$-calculus
with respect to equalities.

\begin{proposition}
 For $\lambda\modal$-terms $M$ and $N$,
 $M \eql[\const{v}] N$ holds
 if and only if $M \eql[\const{c}] N$ holds.
\end{proposition}

We can show the strong normalization theorem of
the computational $\lambda\modal$-calculus
via the strong normalizability of
the $\lambda_{\textrm{c}}$-calculus.

\begin{proposition}
 The computational $\lambda\modal$-calculus is strongly normalizable
 with respect to $\reducto[\const{c}]$\@.
\end{proposition}

\begin{proof}
 Define $\floor{\blank}$ into the typed $\lambda_{\textrm{c}}$-calculus by
 \begin{align*}
  \floor{\boxin{\vect{x}}{\vect{V}}{M}}
  &= \floor{M}\subst{\vect{x}}{\floor{\vect{V}}} ,\\
  \longeqn[\floor{\boxin{\vect{x}}{\vect{V}}{M}}]
  {\floor{\boxin{\vect{w},x,\vect{z}}{\vect{V},A,\vect{N}}{M}}} \\
  &= \letin{y}{\floor{A}}
  {\floor{\boxin{\vect{w},x,\vect{z}}{\vect{V},y,\vect{N}}{M}}} .
 \end{align*}
 One can see that $\floor{V}$ is a value when $V$ is a value,
 remembering that boxed terms are restricted to the form
 $\boxin{\vect{x}}{\vect{M}}{V}$\@.
 Let $\reducto[\ctext{let}_{\modal}]$ be
 the special case of $\reducto[\ctext{let}]$:
 \begin{align*}
  &\boxin{\vect{w},x,\vect{z}}{\vect{W},A,\vect{P}}{M} \\
  &\linecont \reducto[\ctext{let}_{\modal}]
  \letin{y}{A}{\boxin{\vect{w},x,\vect{z}}{\vect{W},y,\vect{P}}{M}}
  &&\norm{\vect{w}} = \norm{\vect{W}} .
 \end{align*}
 It can be checked that
 $M \reducto[\idr{\modal},{\betar[v]{\modal}},\ctext{let}_{\modal}] N$
 implies $\floor{M} \reducto[\const{c}]^\ast \floor{N}$,
 otherwise, $M \reducto[\const{c}] N$ implies
 $\floor{M} \reducto[\const{c}] \floor{N}$\@.
 Because we know the $\lambda_{\textrm{c}}$-calculus is SN
 (it was proved by Hasegawa in \cite{Hasegawa03:STCLCISN}),
 it is sufficient to show
 there is no infinite sequence that consists of
 $\reducto[\idr{\modal}]$, $\reducto[{\betar[v]{\modal}}]$,
 and $\reducto[\ctext{let}_{\modal}]$\@.

 We extend the transformation $\ceil{\blank}$,
 which is defined in the proof of Proposition~2,
 to the computational $\lambda\modal$-calculus by
 \begin{align*}
  \ceil{\letin{x}{N}{M}}
  &= \ceil{M}\subst{x}{\ceil{N}} .
 \end{align*}
 Then, $M \reducto[\idr{\modal},{\betar[v]{\modal}}] N$ implies
 $\ceil{M} \reducto[\betar{\arrow},\etar{\arrow}]^+ \ceil{N}$,
 and $M \reducto[\ctext{let}_{\modal}] N$ implies
 $\ceil{M} \equiv \ceil{N}$\@.
 Suppose the existence of an infinite sequence of
 $\reducto[\idr{\modal}]$, $\reducto[{\betar[v]{\modal}}]$,
 and $\reducto[\ctext{let}_{\modal}]$\@.
 Since there is no infinite reduction sequence
 in the simply typed $\lambda$-calculus,
 neither $\reducto[\idr{\modal}]$ nor $\reducto[{\betar[v]{\modal}}]$
 appears infinitely in the sequence.
 The assumption contradicts the fact that
 there is no infinite sequence of $\reducto[\ctext{let}_{\modal}]$\@.
\end{proof}

\begin{proposition}
 $\reducto[\const{c}]$ is confluent.
\end{proposition}

\begin{proof}
 According to Newman's lemma~\cite{Newman42:TCDE},
 we consider the local confluency.
 Because the $\lambda_{\textrm{c}}$-calculus and
 the call-by-name $\lambda\modal$-calculus are confluent,
 the following critical pairs are essential:
 \begin{align*}
  &\linecont \ureducto[\ctext{let}]
  \letin{y}{M}{\boxin{x}{y}{x}} \\
  &\boxin{x}{M}{x} \\
  &\linecont \dreducto[{\idr{\modal}}] M ,\\
  &\linecont \ureducto[\ctext{let}]
  \letin{z}{(\boxin{\vect{y}}{\vect{N}}{V})}{\boxin{x}{z}{M}} \\
  &\boxin{x}{\boxin{\vect{y}}{\vect{N}}{V}}{M} \\
  &\linecont \dreducto[{\betar[v]{\modal}}]
  \boxin{\vect{y}}{\vect{N}}{M\subst{x}{V}} .
 \end{align*}
 Confluence of the former pair is easily shown.
 For the latter case, let
 $\letin{\vect{w}}{\vect{N'}}{\boxin{\vect{y}}{\vect{W}}{x}}$
 be the $\reducto[\ctext{let}]$-normal form of
 $\boxin{\vect{y}}{\vect{N}}{x}$\@.
 \begin{align*}
  &\lefteqn{\letin{z}{(\boxin{\vect{y}}{\vect{N}}{V})}{\boxin{x}{z}{M}}} \\
  &\reducto[\ctext{let}]^\ast
  \letin{z}{(\letin{\vect{w}}{\vect{N'}}{\boxin{\vect{y}}{\vect{W}}{V}})}
  {\boxin{x}{z}{M}} \\
  &\reducto[\ctext{comp}]^\ast
  \letin{\vect{w},z}{\vect{N'},(\boxin{\vect{y}}{\vect{W}}{V})}
  {\boxin{x}{z}{M}} \\
  &\reducto[{\betar[v]{\ctext{let}}}]
  \letin{\vect{w}}{\vect{N'}}{\boxin{x}{\boxin{\vect{y}}{\vect{W}}{V}}{M}} \\
  &\reducto[{\betar[v]{\modal}}]
  \letin{\vect{w}}{\vect{N'}}{\boxin{\vect{y}}{\vect{W}}{M\subst{x}{V}}} .
  \end{align*}
 On the other hand, the lower term goes to
 the same term by $\reducto[\ctext{let}]^\ast$\@.
\end{proof}

We have restricted forms of terms in
the call-by-value calculi for CPS completeness.
Leaving completeness on one side,
now we can present another CPS transformation on full terms:
\begin{align*}
 \tcpsx{\modal{\sigma}}
 &= \modal((\tcpsx{\sigma}\imply\return)\imply\return) ,\\
 \longeqn[\tcpsx{\modal{\sigma}}]
 {\cpsx{\boxin{\vect{x}}{\vect{N}}{M}}} \\
 &= \labst{k}{\cpsx{\vect{N}}(\labst{\vect{y}}
 {k(\boxin{\vect{z}}{\vect{y}}
 {\labst{h}{\vect{z}(\labst{\vect{x}}{\cpsx{M}h})}})})} ,
\end{align*}
where a non-overridden part of the definition is
just the same as Figure~\ref{FIG:cps}\@.
We remark that the definition of $\cpsx{\blank}$ does not require
a value transformation like $\cpsv{\blank}$\@.
Also this transformation preserves the equality.

\begin{theorem}
 For $\lambda\modal$-terms $M$ and $N$,
 $M \eql[\const{v}] N$ implies $\cpsx{M} \eql[\const{n}] \cpsx{N}$\@.
\end{theorem}

Unfortunately, it can be seen that
this modified CPS transformation does not reflect the equality.
For example,
\begin{align*}
 &\cpsx{\boxin{x}{\boxin{y}{L}{N}}{M}} \\
 &\linecont \eql[\const{n}]
 \cpsx{\boxin{y}{L}{(\labst{x}{M})N}} ,
\end{align*}
but $\boxin{x}{\boxin{y}{L}{N}}{M}
\not\eql[\const{v}] \boxin{y}{L}{(\labst{x}{M}){N}}$
unless $N$ is a value.
It is still open to find an axiomatization
complete for $\cpsx{\blank}$\@.

\section{Semantics}\label{SS:sem}

Since Kripke semantics~\cite{Kripke63:SAMLINPL}
concern only provability,
they are not suitable for our study.
It is proposed by Bellin et al.\ in \cite{BellinPaiva01:ECHCBCML}
that a model of \logic{IK} is
a cartesian closed category with a lax monoidal endofunctor
with respect to cartesian products.
(Fundamental properties of monoidal functors
are found in \cite{MacLane97:CWM}\@.)
Indeed, it is shown in \cite{Kakutani07:CNCVNML} that
the call-by-name $\lambda\modal$-calculus with conjunctions
is sound and complete for the class of such models.
The completeness without conjunctions is expected to be proved
in a way similar to the case of the simply typed $\lambda$-calculus.
Bellin et al.'s calculus has the same syntax as ours,
but it is not complete for the semantics.

Semantics for the call-by-value calculus is
more complex than the call-by-name semantics.
We show construction of a call-by-value model as follows.

Let a cartesian closed category $\mathcal{C}$ have
a strong monad $\tuple{T,\unit,\mult}$ and
a monoidal endofunctor $\tuple{\modal,\const[1]{m},\const{m}}$\@.
We focus on the Kleisli category $\mathcal{C}_{T}$,
which is a model of the $\lambda_{\textrm{c}}$-calculus.
For a morphism $f \in \mor[\mathcal{C}]{B}{A}$,
there exists a morphism
$\unit\comp\modal{f} \in \mor[\mathcal{C}_{T}]{\modal{B}}{\modal{A}}$\@.
This fact explains a construction
\infrule{\var{x}{\sigma} \prove V \type{\tau}}
{\var{y}{\modal{\sigma}} \prove \boxin{x}{y}{V} \type{\modal{\tau}}}
which is functorial:
\begin{align*}
 &\boxin{x}{M}{x} \eql[\const{v}] M ,\\
 &\boxin{x}{\boxin{y}{M}{W}}{V}
 \eql[\const{v}] \boxin{y}{M}{V\subst{x}{W}} .
\end{align*}
The natural transformation
$\{ \const[A,B]{m} \in
\mor[\mathcal{C}]{\modal{A}\times \modal{B}}{\modal(A\times B)} \}$
induces a type-indexed family
$\{ \unit\comp\const[A,B]{m} \in
\mor[\mathcal{C}_{T}]{\modal{A}\times \modal{B}}{\modal(A\times B)} \}$\@.
This family is not a natural transformation but natural in values.
It explains an equation
\begin{align*}
 &\boxin{x,z}{\boxin{y}{N}{V},P}{M}
 \eql[\const{v}] \boxin{y,z}{N,P}{M\subst{x}{V}} .
\end{align*}

If a monad $T$ is a continuation monad,
that is, $TX = \return^{\return^{X}}$,
the categorical semantics coincides with the CPS semantics.

\section{Extensions}\label{SS:S4}

In this section, we show an extension of
the call-by-name calculus to \logic{IS4}\@.
A call-by-value axiomatization still remains future work.

We introduce type-indexed families of constants
$\{ \counit[\sigma] \type{\modal{\sigma}\imply\sigma} \}$ and
$\{ \comult[\sigma] \type{\modal{\sigma}\imply\modal{\modal{\sigma}}} \}$
with the following axioms:
\begin{align*}
 &\counit(\boxin{\vect{x}}{\vect{N}}{M})
 \eql[\rulename{nat}_{\counit}] M\subst{\vect{x}}{\counit\vect{N}} ,\\
 &\comult(\boxin{\vect{x}}{\vect{N}}{M})
 \eql[\rulename{nat}_{\comult}] \boxin{\vect{y}}{\comult\vect{N}}
 {\boxin{\vect{x}}{\vect{y}}{M}} ,\\
 &\comult(\comult M)
 \eql[\ctext{mon}] \boxin{x}{\comult M}{\comult x} ,\\
 &\counit(\comult M)
 \eql[\rulename{mon}] \boxin{x}{\comult M}{\counit x}
 \eql[\rulename{mon}] M .
\end{align*}
(It is trivial that this calculus corresponds to \logic{IS4}\@.)
We only consider equalities
because it is not obvious in some equations
which side is a result of a computation.
Naturally, it is possible to give calculi
for \logic{IT} and \logic{IK4}
as fragments of this \logic{IS4} calculus.

Bierman and de~Paiva introduced
the $\lambda^{\logic{S4}}$-calculus
in \cite{BiermanPaiva00:IML}\@.
We show our calculus can emulate the $\modal$-fragment
of their calculus.
Let
\begin{align*}
 &\boxsub{\vect{x}}{\vect{N}}{M}
 \equiv \boxin{\vect{x}}{\comult\vect{N}}{M} ,\\
 &\unbox{M} \equiv \counit M .
\end{align*}
The following equation holds in our calculus:
\begin{align*}
 &\unbox{(\boxsub{\vect{x}}{\vect{N}}{M})}
 \eql M\subst{\vect{x}}{\vect{N}} .
\end{align*}
On the other hand, Bierman and de~Paiva's calculus does not
emulate our calculus
because ours is complete for the class of
cartesian closed categories with monoidal comonads
but theirs is not.

It is also possible to compare our calculus
to a dual context version of \logic{IS4}
like Barber and Plotkin's DILL~\cite{Barber96:DILL}\@.
A dual context calculus for \logic{IS4} is
proposed by Pfenning and Davies in \cite{PfenningDavies01:JRML}\@.
(Their calculus has a diamond modality too,
but we just ignore it here.)
The dual context calculus requires
new syntax $\modal{M}$ and $\letinmodal{x}{N}{M}$
instead of $\boxin{\vect{x}}{\vect{N}}{M}$\@.
The typing rules consist of
\infrule{}{\varDelta,\var{a}{\tau},\varDelta' \csep \varGamma
 \prove a \type{\tau}}
\infrule{\varDelta \csep {} \prove M \type{\tau}}
{\varDelta \csep \varGamma \prove \modal{M} \type{\modal{\tau}}}
\infrule{\varDelta,\var{a}{\sigma} \csep \varGamma \prove M \type{\tau}
 \andalso \varDelta \csep \varGamma \prove N \type{\modal{\sigma}}}
{\varDelta \csep \varGamma \prove \letinmodal{\uptype{\sigma}a}{N}{M} \type{\tau}}
and the usual rule of the simply typed $\lambda$-calculus
with respect to right-hand contexts.
We use $a$, $b$, \dots for variables of left-hand contexts
to distinguish them from those of right-hand contexts.
The equality is defined by
\begin{align*}
 &C \type{\text{context}} ,\\
 &\letinmodal{a}{\modal{N}}{M} \eql M\subst{a}{N} ,\\
 &\letinmodal{a}{M}{\modal{a}} \eql M ,\\
 &C[\letinmodal{a}{N}{M}] \eql \letinmodal{a}{N}{C[M]}
 &&a \not\in \freev{C} ,
\end{align*}
where $C$ is a context such that
its hole does not appear under a box.

In the $\lambda\modal$-calculus,
we call the following equality the strongness condition:
\begin{align*}
 &\boxin{\vect{w},x,\vect{z}}{\vect{P},N,\vect{Q}}{M} \\
 &\linecont \eql[\rulename{st}] \boxin{\vect{w},\vect{z}}{\vect{P},\vect{Q}}{M}
 &&\norm{\vect{w}} = \norm{\vect{P}} ,\\
 &\boxin{\vect{w},x,y,\vect{z}}{\vect{P},N,N,\vect{Q}}{M} \\
 &\linecont \eql[\rulename{st}]
 \boxin{\vect{w},x,\vect{z}}{\vect{P},N,\vect{Q}}{M\subst{y}{x}}
 &&\norm{\vect{w}} = \norm{\vect{P}} .
\end{align*}
The following equality is called the symmetricity.
\begin{align*}
 &\boxin{\vect{w},x,y,\vect{z}}{\vect{P},N,L,\vect{Q}}{M} \\
 &\linecont \eql[\rulename{sym}]
 \boxin{\vect{w},y,x,\vect{z}}{\vect{P},L,N,\vect{Q}}{M}
 &&\norm{\vect{w}} = \norm{\vect{P}} .
\end{align*}
The terms ``strong'' and ``symmetric'' follow
the terms ``strong monoidal functor'' and
``symmetric monoidal functor'' in the category theory.
We show that the dual context calculus is equivalent to
the $\lambda\modal$-calculus with
the symmetricity and the strongness condition.
Define the transformation $\ceilx{\blank}$ from
the dual context calculus into the $\lambda\modal$-calculus by
\begin{align*}
 \ceilx{a} &= \counit a ,\\
 \ceilx{\modal{M}}
 &= \boxin{\vect{b}}{\comult{\vect{b}}}{\ceilx{M}}
 &&\mathconj{\text{where}} \{\vect{b}\} = \freev{M} ,\\
 \ceilx{\letinmodal{a}{N}{M}}
 &= (\labst{a}{\ceilx{M}})\ceilx{N} .
\end{align*}
For a derivable judgment
$\var{a_1}{\rho_1},\ldots \csep \var{x_1}{\sigma_1},\ldots
\prove M \type{\tau}$,
the judgment
\begin{align*}
 &\var{a_1}{\modal{\rho_1}},\ldots,\var{x_1}{\sigma_1},\ldots
 \prove \ceilx{M} \type{\tau}
\end{align*}
is derivable in the $\lambda\modal$-calculus,
and $M \eql N$ implies $\ceilx{M} \eql \ceilx{N}$
under the strongness condition and the symmetricity.
Its inverse $\floorx{\blank}$ can be defined by
\begin{align*}
 \floorx{\counit}
 &= \labst{y}{\letinmodal{a}{y}{a}} ,\\
 \floorx{\comult} &= \labst{y}{\letinmodal{a}{y}{\modal{\modal{a}}}} ,\\
 \floorx{\boxin{\vect{x}}{\vect{N}}{M}}
 &= \letinmodal{\vect{a}}{\floorx{\vect{N}}}
 {\modal(\floorx{M}\subst{\vect{x}}{\vect{a}})} .
\end{align*}
One can see that $M \eql N$ implies $\floorx{M} \eql \floorx{N}$\@.
While $\ceilx{\floorx{\blank}}$ is the identity up to the equality,
$\floorx{\ceilx{\blank}}$ is not the identity itself.
The reason is that the dual context calculus is
redundant in some sense:
for example, two judgments,
\begin{align*}
 &{} \csep \var{x}{\modal{\sigma}} \prove x \type{\modal{\sigma}} ,\\
 &\var{a}{\sigma} \csep {} \prove \modal{a} \type{\modal{\sigma}} ,
\end{align*}
have the same semantics.

Another possible extension of our calculus is
a calculus corresponding to the classical modal logic \logic{K}\@.
In \cite{Parigot92:LMC}, Parigot has extended
the simply typed $\lambda$-calculus to the $\lambda\mu$-calculus,
which corresponds to the classical logic.
We can extend the call-by-name $\lambda\modal$-calculus
with $\mu$-operator in a straightforward way.
A call-by-value version and
analyses of the relation between call-by-name and call-by-value
are found in \cite{Kakutani07:CNCVNML}\@.

\section*{Acknowledgments}

I am grateful to Tatsuya Abe.
This work has developed over discussions with him.
I also thank PPL referees for valuable comments.

\bibliographystyle{plain}
\bibliography{valkyrie}

\end{document}